\documentclass[12pt]{article}

\usepackage[left=1.25in,top=1.0in,right=1.25in,bottom=1.0in]{geometry}
\usepackage{amsfonts}
\usepackage{amssymb}
\usepackage{graphicx}
\usepackage{setspace}
\usepackage[round]{natbib}
\usepackage{amsmath}
\usepackage{amssymb}
\usepackage{amsbsy}
\usepackage{amsthm}
\usepackage{paralist}
\usepackage{bm}
\usepackage{booktabs}
\usepackage{epsfig}
\usepackage[plain,noend]{algorithm2e}

\newcommand{\N}{\mbox{N}}

\newcommand{\one}{1}
\def\half{\hbox{$1\over2$}}
\newcommand{\xtil}{\tilde{x}}

\newtheorem{theorem}{Theorem}
\newtheorem{proposition}[theorem]{Proposition}

\title{Nonparametric Bayesian testing for monotonicity}

\author{James G.~Scott, Thomas S.~Shively\footnote{McCombs School of Business, University of Texas at Austin; 2110 Speedway B6500, Austin, Texas, USA 78712; james.scott@mccombs.utexas.edu; tom.shively@mccombs.utexas.edu}
\\
and Stephen G.~Walker\footnote{Department of Mathematics, University of Texas at Austin; Austin, Texas, USA 78712; s.g.walker@math.utexas.edu  }
}

\date{April 2014}

\begin{document}

\begin{spacing}{1.1}

\maketitle

\begin{abstract}
This paper studies the problem of testing whether a function is monotone from a nonparametric Bayesian perspective.  Two new families of tests are constructed. The first uses constrained smoothing splines, together with a hierarchical stochastic-process prior that explicitly controls the prior probability of monotonicity. The second uses regression splines, together with two proposals for the prior over the regression coefficients. The finite-sample performance of the tests is shown via simulation to improve upon existing frequentist and Bayesian methods.  The asymptotic properties of the Bayes factor for comparing monotone versus non-monotone regression functions in a Gaussian model are also studied. Our results significantly extend those currently available, which chiefly focus on determining the dimension of a parametric linear model.

\vspace{0.5\baselineskip}


\end{abstract}

\newpage

\section{Introduction}

Many authors have used Bayesian methods for estimating functions that are known to have certain shape restrictions, including positivity, monotonicity, and convexity. Examples of work in this area include \citet{holmes:heard:2003}, \citet{neelon:dunson:2004}, \citet{dunson:2005}, \citet{shively:sager:walker:2009}, \citet{shively:walker:damien:2011}, and \citet{hannah:dunson:2012b}.  But an important question unaddressed by this work is whether it is appropriate to impose a specific shape constraint.  If it is, the resulting function estimates are often considerably better than those obtained using unconstrained methods \citep{cai:dunson:2007}. Conversely, the inappropriate imposition of constraints will result in poor function estimates.

This paper considers the use of nonparametric Bayesian methods to test for monotonicity.  The approach has connections with two major areas of the Bayesian literature.  First, we extend work on regression and smoothing splines \citep{smith:kohn:1996,dimatteo:etal:2001,pan:smith:2008,shively:sager:walker:2009}, and dictionary expansions more generally \citep{clyde:wolpert:2007}, adapting these tools in such a way that they become appropriate for hypothesis testing.  Second, we extend work on default Bayesian model selection \citep{bergerpericchi2001,georgefoster2000,clyde:george:2004,giron:etal:2010} to the realm of nonparametric priors for functional data.

Our main theoretical contribution is to characterize the asymptotic properties of the Bayes factor for comparing monotone versus non-monotone mean regression functions in a normal model. To date, such asymptotic analysis has focused on comparing normal linear models of differing dimension; see, for example, \citet{giron:etal:2010}.  We generalize this work to the problem of comparing non-linear mean functions.  We do so using a technique introduced in \citet{walker:hjort:2002} that combines the properties of maximum-likelihood estimators and Bayesian marginal likelihoods that involve the square roots of likelihood functions.

Our main methodological contribution is to construct two new families of priors that are appropriate in a Bayesian test for monotonicity.  The first approach uses a monotone version of smoothing splines, together with a hierarchical stochastic-process prior that allows explicit control over the time at which the underlying function's derivative first becomes negative. The second approach constructs a family of tests using a regression spline model \citep{smith:kohn:1996} with a mixture of constrained multivariate distributions as the prior for the regression coefficients.  We study two possible choices for this prior: constrained multivariate Gaussian distributions, and a constrained version of the multivariate non-local method-of-moments distribution \citep{johnson:rossell:2010}. The mixing parameters in both priors can be adjusted to allow the user to set the prior probability of a monotone function.

These approaches work for a very general class of sampling models, including those for continuous, binary, and count data.  They are most easily understood, however, in the case where data $Y = \{y_1, \ldots, y_n\}$, observed at ordered points $x \in \{x_1, \ldots, x_n\}$, are assumed to arise from a normal sampling model, $(y_i \mid f_i) \sim \N(f_i, \sigma^2)$.  Here $f_i = f(x_i)$ denotes the value at $x_i$  of an unobserved stochastic process $f(x)$, with $x \in \mathcal{X}$ and $f$ in a space  $\mathcal{F}$ of real-valued functions on $\mathcal{X}$.  Without loss of generality, we consider functions defined on $\mathcal{X} = [0,1]$. To test whether $f(x)$ is non-decreasing, one must compare the evidence for $H_1: f \in \mathcal{F}_1$ against $H_2: f \in \mathcal{F}_2$, where
\begin{eqnarray*}
\mathcal{F}_1 &=& \{f \in \mathcal{F} :  f(s_2) \geq f(s_1) \; \mbox{for all pairs} \; s_2 \geq s_1 \} \, ,\\
\mathcal{F}_2 &=& \{f \in \mathcal{F} :  f(s_2) < f(s_1) \;  \mbox{for at least one pair} \; s_2 \geq s_1 \} \, .
\end{eqnarray*}

Under the Bayesian approach to this problem, two important goals are in conflict.  One is flexibility: we wish to make few assumptions about $f(x)$, so that the testing procedure may accommodate a wide class of functions.  Applying this principle na\"ively would lead us to choose an encompassing function space $\mathcal{F}$ with large support, and vague prior measures $\Pi_1$ and $\Pi_2$ over $\mathcal{F}_1$ and $\mathcal{F}_2$, respectively.  On the other hand, the Bayes factor for comparing $H)1$ versus $H_2$ is
$$
\mbox{BF}(H_1 : H_2) = \frac{\int_{\mathcal{F}_1} \left\{ \prod_{i=1}^N \phi(y_i \mid f(x_i), \sigma^2) \right\} \ d \Pi_1(f)  }{\int_{\mathcal{F}_2} \left\{ \prod_{i=1}^N \phi(y_i \mid f(x_i), \sigma^2) \right\} \ d \Pi_2(f)  } \, ,
$$
which is heavily influenced by the dispersion of $\Pi_1$ and $\Pi_2$ \citep{Scott:2008}.  This strongly contra-indicates the use of noninformative priors for model selection.  In particular, improper priors may not be used, as this leaves the Bayes factor defined only up to an arbitrary multiplicative constant.

To see how these conflicting goals may be balanced, suppose initially that $\mathcal{F} = \mathcal{C}^1$, the space of all continuously differentiable functions.  Membership in $\mathcal{C}^1$ is easily enforced by supposing that $f(x) = \int_0^x g(s) \ ds$, where $g(s) \in \mathcal{C}$, the space of real-valued stochastic processes with almost-surely continuous sample paths.  The test may now be phrased in terms of $H_1: g(s) \in C^{+}$ versus $H_2: g(s) \in C^{-}$, where
$$
\mathcal{C}^+ = \{g \in \mathcal{C} :  g(s) > 0 \; \mbox{for all} \; s \in [0,1] \} \, , \quad \mathcal{C}^- =  \mathcal{C} \ \backslash \ \mathcal{C}^+  \, .
$$
The prior for $g(s)$ now determines the behavior of the resulting estimate.

There are two similar proposals to ours in the Bayesian literature on shape-constrained inference.  First, \citet{dunson:2005} used a Bonferroni-like method for controlling the overall probability of monotonicity via a semiparametric model akin to a variable-selection prior.  Second, \citet{salomond:2014} introduces an adaptive Bayes test for monotonicity based on locally constant functions.  There is also a large body of classical work on specification tests, many of which are appropriate for assessing monotocity or more general aspects of functional form.  We cite four examples here, which collectively represent a broad range of ideas about how to approach the testing problem; for a comprehensive bibliography, see \citet{akakpo:2014}.  \citet{zheng:1996} proposed a test of functional form based on the theory of U-statistics.  \citet{bowman:etal:1998} used a two-stage boostrap test for monotonicity based on the idea of a critical bandwidth parameter from an unconstrained estimate.  \citet{baraud:2005} devise a test based on comparing local means of consecutive blocks of observations.  \citet{akakpo:2014} propose a method based on the empirical distance between the cumulative sum of the data points and the least concave majorant of this running sum.  Our simulation study finds that the new approaches proposed here generally outperform these existing Bayesian and frequentist methods, in that they generally have better power to detect departures from monotonicity at a fixed false-positive rate.

\section{Constrained smoothing splines}
\label{sec:constrainedBM}

A straightforward nonparametric Bayesian strategy to estimate a continuous function $f(x)$ is to model the derivative $g(x) = f'(x)$ as a scaled Wiener process, implying that the increments $g(s_2) - g(s_1)$ are normally distributed with mean zero and variance $\tau^2 (s_2 - s_1)$.  But this is inappropriate for model selection, since it places unit prior probability on the hypothesis of non-monotonicity.  To see this, recall that for every $c>0$, the sample path of a Wiener process almost surely takes both positive and negative values on $(0,c)$, implying that $f$ will not be monotone.

However, we can adapt the smoothing-spline approach to the testing problem as follows.  Introduce $\xi = \inf_{s} \{ s: g(s)=0 \}$ as a latent variable that denotes the first downward crossing point of $g(x)$.  This maps directly onto the hypotheses of interest: $\xi \geq 1$ if and only if the derivative $g(s)$ is strictly positive on the unit interval.

The introduction of $\xi$ turns a hypothesis test for an infinite-dimensional object $f$ into a one-dimensional estimation problem, via a hierarchical model for $f(x)$:
\begin{eqnarray}
(y_i \mid f) &\sim& \N(f(x_i), \sigma^2) \,  , \quad f(x) = \int_0^{x} g(s) \ ds \, ,\nonumber\\
(g \mid \xi, \tau) &\sim& \Pi(g \mid \xi, \tau) \,  , \quad \xi \sim p(\xi) \, .
\end{eqnarray}
Here $\Pi(g \mid \xi, \tau)$ denotes the conditional probability distribution of a Wiener process with scale parameter $\tau$, given $\xi = \inf_{s} \{ s: g(s)=0 \}$.  This differs from the one-dimensional Brownian bridge, in that $g(x)$ is restricted to be positive over the interval $(0, \xi)$.

We refer to the overall approach as a constrained smoothing spline.  To test $H_1$ versus $H_2$, we compute the posterior probability of $H_1$ as $\hat p_1 = \mbox{pr}(\xi > 1  \mid Y)$.  Note that $\hat p_1/(1-\hat p_2)$ is identical to the Bayes factor in favor of $H_1$, as long as the prior distribution for $\xi$ has exactly half its mass on $(1, \infty)$.  One default choice for $p(\xi)$ for which this holds is a standard log-normal distribution, although in principle, any proper prior with a median of 1 could be used.   Moreover, even if $\mbox{pr}(\xi > 1) \neq 0.5$ a priori, it is easy to compute the Bayes factor by simply dividing $\hat p _1/(1-\hat p _1)$ by the prior odds ratio in favor of monotonicity.

As the data are observed on a discrete grid, it is necessary to characterize the increments of the above stochastic-process prior.  The following proposition describes the distribution of these increments, here referred to as the fractional normal distribution.  Our proof generalizes the argument of \citet{chigansky:klebaner:2008} to a scaled Wiener process.

\begin{proposition}
\label{prop:fracnorm}
Let $\Pi(g \mid \xi, \tau)$ denote the conditional probability measure of a scaled Wiener process $g(s)$ with scale parameter $\tau$, given the condition that the first downward crossing of $g(x)$ occurs at $x = \xi$, where conditioning is meant in the sense of Doob's $h$ transform.  Let $0 < s_1 < \xi$, and define $g_1 \equiv g(s_1)$ and $g_0 \equiv g(0)$.  The conditional distribution $p(g_1 \mid g_0, \xi, \tau)$ arising from $\Pi(g \mid \xi, \tau)$ is a fractional
normal distribution, denoted $\mbox{FN}(g_1 \mid g_0, \xi, \tau)$, with density
\begin{equation}
\label{fracnormdens}
p(g_1 \mid g_0, \xi, \tau) = \frac{\sqrt{2} e^{-m^2/2}}{mh^2\sqrt{\pi}} \ g_1 \sinh(m g_1/h) \ \exp\left\{- \frac{g_1^2}{2h^2}\right\} \, ,
\end{equation}
where $h = \tau \{\xi u (1-u)\}^{1/2}$ and $m = (g_0/\tau) \{(1-u)/(\xi u)\}^{1/2}$.
\end{proposition}

\begin{proof}
Let $V$ be a three-dimensional Brownian bridge with scale $\tau$, with $V(0) = v$ and $V(\xi) = 0$, for some point $v$ having Euclidean norm $\Vert v \Vert = g_0$.  That is,
\begin{equation}
\label{eqn:3dbrownbridge}
V(s) = v + W(s)  - \frac{s}{\xi} \{ W(\xi) + v \} \quad \mbox{for} \quad s < \xi \, ,
\end{equation}
where $W(s)$ is a Wiener process in three dimensions with scale parameter $\tau$.  Observe that the random variable of interest is equal in distribution to the radial part of $V$, observed at $s_1$:
$$
(g_1 \mid \xi, g_0) \stackrel{D}{=} \Vert V(s_1) \Vert \, .
$$
Moreover, since $\Vert V \Vert$ is independent of the the starting point $v$, as long as $\Vert v \Vert = g_0$, one may choose $v = (g_0, 0, 0)$ to simplify the calculations.  From (\ref{eqn:3dbrownbridge}), this leads to
\begin{equation}
\label{eqn:bbridgenorm}
\Vert V(s_1) \Vert \stackrel{D}{=} \left\{   \left(  z_1 \tau \sqrt{ u (1-u) \xi } +g_1(1-u) \right)^2 + \xi u (1-u) \tau^2 z_2^2 + \xi u (1-u) \tau^2 z_3^2    \right\}^{1/2} \, ,
\end{equation}
where $u = (\xi-s_1)/(\xi-s_0)$, and where $(z_1, z_2, z_3)$ are independent, standard-normal draws.  Equivalently,
$$
\Vert V(s_1) \Vert \stackrel{D}{=} h \left\{   \left( z_1 + m \right)^2 +  z_2^2 + z_3^3    \right\}^{1/2} \, ,
$$
with $m$ and $h$ defined as above.  The density $z_2^2 + z_3^2$ is that of an exponential distribution with rate $1/2$.  Meanwhile, the density of the term $\eta = (z_1 + m)^2$  may be evaluated directly using the fact that $p(\eta) = \frac{d}{du} P(\eta < u)$ and differentiating under the integral sign.  After a simple change of variables, the density of $\theta = \left\{   \left( z_1 + m \right)^2 +  z_2^2 + z_3^3    \right\}^{1/2}$ may be computed via convolution as
$$
f_c(\theta) = \frac{\sqrt{2} e^{-m^2/2}}{m\sqrt{\pi}} \ \theta \sinh(m \theta) \ \exp(-\theta^2/2) \, .
$$
The result follows from the fact that the density of $g_1 = h \theta$ is $p(g_1) = f_c(\theta/h)/h$.
\end{proof}

Proposition \ref{prop:fracnorm} characterizes the distribution of the increments of $g(x)$ for a given $\xi$.  This leads to an efficient sequential Monte Carlo algorithm for fitting the implied state-space model for the unknown function $f(x)$ and crossing time $\xi$.  The details of the algorithm are in the technical appendix.

\section{Constrained regression splines}

This section develops a second family of tests for monotonicity using a regression spline model with non-standard prior distributions, calibrated specifically to the testing problem.  In principle, regression splines of any order may be used.  In practice, we use quadratic splines, because they lead to excellent finite-sample performance and a tractable sampling scheme.  Thus our finitely parametrized approximation to the function $f(x)$ is given by the expansion
\begin{equation}
\label{eqn:splinemodel}
f_m(x) = \alpha + \beta_1 x + \beta_2 x^2 + \beta_3 (x - \xtil_1)^2_+ +  \cdots + \beta_{m+2}(x -\xtil_{m})^2_+ \, ,
\end{equation}
where the $\xtil_j$ are the ordered knot locations, and $z_+ = \max(z, 0)$.   Throughout this section we assume the model $y_i = f_m(x_i) + \epsilon_i$ for some choice of $m$.  For notational convenience let $\xtil_0 = 0$ and $\xtil_{m+1} = 1$.   Clearly the first derivative $f'_m(x)$ is linear between each pair of knots.  Therefore, if $f'_m(\xtil_j) \geq 0$ and $f'_m(\xtil_{j+1}) \geq 0$, then $f'_m(x) \geq 0$ on $[\xtil_j, \xtil_{j+1}]$.  If this condition holds for all $j = 0, \ldots, m$, then $f_m(x)$ is monotone on $[0,1]$.

To develop a test for monotonicity based on regression splines, we generalize \cite{smith:kohn:1996} and \citet{shively:sager:walker:2009}, using Bayesian variable selection to choose the location of the knots from a set of $m$ pre-specified values.  Using all $m$ knots, (\ref{eqn:splinemodel}) may be re-written in matrix notation as $y = \alpha 1 + X \beta + \epsilon$, where $y$ is the $n\times 1$ vector of observations, $1$ is a vector of ones, $X$ is an $n \times (m+2)$ design matrix, and $\beta$ is the vector of spline coefficients.  Let $\iota$ be a vector of indicator variables whose $j$th element takes the value 0 if $\beta_j = 0$, and 1 otherwise.  Let $\beta_{\iota}$ consist of the elements of $\beta$ corresponding to those elements of $\iota$ that are equal to one, and let $p = |\iota|$ denote the number of nonzero entries in $\iota$.   \citet{shively:sager:walker:2009} derive the constraints on $\beta_{\iota}$ that ensure the monotonicity of $f_m(x)$ for a given $\iota$.  Specifically, $f_m(x)$ is monotone whenever the vector $\gamma_{\iota} \equiv L_{\iota} \beta_{\iota} \geq 0$, where $L_\iota$ is a known lower-triangular matrix that depends on $\iota$ and the $\xtil_j$'s.  Checking whether any element of $\gamma_{\iota}$ is negative, and controlling the prior probability of such a constraint violation across all possible values of $\iota$, are the basis of our test for monotonity.  The matrix $L_{\iota}$ therefore plays an important role in the prior distribution for $\beta_{\iota}$.

Given $\iota$, the $\beta_\iota$ space is divided into $2^p$ disjoint regions denoted $R_\iota^{(1)}, \ldots, R_\iota^{(2^p)}$, with each region defined by a different combination of signs of $f_m'(x)$ at each of the included knots.  Without loss of generality we may let $R^{(1)}_\iota$ denote the region where the derivative is non-negative at each of the included knots, in which case $f'_m(x) \geq 0$ for all $x \in [0,1]$.  For a specific $\iota$ and prior $p(\beta_\iota)$, one may compute the prior probability $\mbox{pr}( \beta_\iota \in R^{(1)}_\iota \mid \iota)$, which is identical to $\mbox{pr}\{ f_m'(x) \geq 0 \ \mbox{for all} \ x \in [0,1] \mid \iota \} $.

The key feature of our approach is that we specify priors on $\iota$ and $\beta_{\iota}$ that allow explicit control over the prior probability of a monotone function.  Given these priors, we compute $\mbox{pr}\{ f_m'(x) \geq 0 \ \mbox{for all} \ x \in [0,1] \mid y \}$, the posterior probability of a monotone function with $\iota$ and $\beta_\iota$ marginalized out.  The function is declared to be monotone if this probability exceeds a specified threshold.

  We now describe the overall approach for constructing $p(\iota)$ and $p(\beta_\iota)$.  The $\iota_j$ are assumed to be independent with $\mbox{pr}(\iota_j = 0) = p_j$.  Given $\iota$ and the error variance $\sigma^2$, the prior for $\beta_\iota$ is a discrete mixture of $2^p$ multivariate distributions,
$$
\beta_{\iota} \sim \sum_{d=1}^{2^p} q_d \Pi_d \, ,
$$
where $\Pi_d$ is constrained to have support on $R^{(d)}_\iota$.  Within each component of the mixture, there is a fixed combination of signs of $f_m'(x)$ at each of the included knots.  As $R^{(1)}_\iota$ corresponds to the region where $f_m'(x) \geq 0$ for all $x$, $q_1$ is the prior probability of monotonicity, conditional on $\iota$.

Two specific choices for $\Pi_d$ are considered: one based on the multivariate Gaussian distribution, and the other based on the multivariate method-of-moments distribution described by \citet{johnson:rossell:2010}.  Both priors involve the constraint matrix $L_\iota$ in order to ensure that the necessary integrals are analytically feasible.  This is analogous to the use of $g$-priors in ordinary linear regression. In the technical appendix these priors are constructed in detail, and a Markov-chain sampling algorithm is presented for sampling from the joint posterior distribution over all model parameters.

\section{Asymptotic properties of Bayes factors}

\subsection{Independent and identically distributed models}

This section develops the asymptotic properties of the Bayes factor in a test for monotonicity. This analysis requires a new approach that significantly extends previous methods used to study Bayes factors in the context of parametric linear models.  To introduce the new approach, we first consider the case of independent and identically distributed observations. This simplified setting allows us to convey the essential idea of the new argument, ignoring many of the additional technical complications that arise in a test for monotonicity of a regression function.

Suppose two Bayesian models, $M_1$ and $M_2$, are to be compared via the Bayes factor:
\begin{equation}
\label{eqn:bayesfactor1}
B_{12}=\frac{\int_{\mathbb {F}_1}\prod_{i=1}^n f_1(x_i \mid \theta_1)\,\pi_1(d\theta_1) }
{\int _{\mathbb{F}_2}\prod_{i=1}^n f_2(x_i \mid \theta_2)\,\pi_2(d\theta_2)} \, .
\end{equation}
Let $d_K(f,g)=\int f\log (f/g)$ be the Kullback--Leibler divergence between $f$ and $g$, and let
$$
d_H(f,g)=\left\{ \int \left(\sqrt{f}-\sqrt{g} \right)^2 \right\}^{1/2}
$$
be the Hellinger distance between $f$ and $g$, which is bounded by 2.

\begin{theorem}
Suppose that the data $x_1, \ldots x_n$ are assumed to arise from some true density $f_0(x)$, and that models $M_1$ and $M_2$ are to be compared, where $M_j=\{f_j(x \mid \theta_j),\,\pi_j(\theta_j),\,\theta_j\in \mathbb{F}_j\}.$  First, suppose that the true $f_0(x)$ is in the Kullback--Leibler support of $\pi_2$: for all $\epsilon>0$,
\begin{equation}
\label{eqn:condition1}
\pi_2 \left[ \theta_2: \,d_{K} \{ f_0(\cdot),f_2(\cdot \mid \theta_2) \} <\epsilon \right] >0 \, .
\end{equation}

Second, suppose that for all sufficiently large $n$ and for any $c>0$, the following bound holds almost surely under $\{f_0(x), \sigma_0\}$:
\begin{equation}
\label{eqn:condition2}
\sup_{\theta_1\in {\mathbb{F}_1} }\prod_{i=1}^n \frac{f_1(x_i \mid \theta_1)}{f_0(x_i)}<e^{nc} \, .
\end{equation}

Finally, suppose that
\begin{equation}
\label{eqn:condition3}
\inf_{\theta_1\in{\mathbb{F}_1}} d_H \left\{ f_1(\cdot \mid \theta_1),f_0(\cdot) \right\} >0 \, .
\end{equation}
Then $B_{12}\rightarrow 0$ almost surely under $f_0$.

\end{theorem}

\begin{proof}
To prove the result, consider the denominator of $B_{12}$ in Equation (\ref{eqn:bayesfactor1}) with a factor introduced, which is also introduced to the numerator, and so the factor cancels out:
$$I_{n2}=\int_{\mathbb{F}_2} \prod_{i=1}^n \frac{f_2(x_i \mid \theta_2)}{f_0(x_i)}\,\pi_2(d\theta_2).$$
It is well known that with the first condition of the theorem, $I_{n2}>e^{-n\tau}$ almost surely for all large $n$ and for any $\tau>0$. See, for example, \citet{schwartz:1965}.

Now consider the numerator with the same factor introduced:
$$I_{n1}=\int_{\mathbb{F}_1} \prod_{i=1}^n \frac{f_1(x_i \mid \theta_1)}{f_0(x_i)}\,\pi_1(d\theta_1).$$
We can write this with an upper bound as follows:
$$I_{n1}\leq \left\{\sup_{\theta_1\in {\mathbb{F}_1} }\prod_{i=1}^n \frac{f_1(x_i \mid \theta_1)}{f_0(x_i)}\right\}^{1/2}
\,\int_{\mathbb{F}_1} \left\{\prod_{i=1}^n \frac{f_1(x_i \mid \theta_1)}{f_0(x_i)}\right\}^{1/2}\,\pi_1(d\theta_1).$$
The second condition will ensure that the first term remains bounded.   The second term, labelled as $J_{n1}$, has expectation
$$
E (J_{n1}) \leq \int_{\mathbb{F}_1} \bigg[1- \half d_H^2\{ f(\cdot \mid \theta_1),f_0(\cdot) \} \bigg]^n\,\pi_1(d\theta_1) \, .
$$
Hence, with the third condition, for some $\eta>0$, $E (J_{n1}) <e^{-n\eta}$ and thus $J_{n1}<e^{-n\delta}$ almost surely for all large $n$ and for some $\delta>0$.

Putting these together, we have
$$B_{12}<\exp\left\{-n(\delta-\half c-\tau)\right\}\quad\mbox{almost surely for all large   }n\mbox{     for any    }c,\tau>0.$$
Choose $\half c+\tau<\delta$ to obtain the desired result.
\end{proof}

Using a similar argument, it is also possible to show that $B_{12}\rightarrow\infty$ almost surely if the following three conditions hold.  First, the true $f_0(x)$ is in the Kullback--Leibler support of $\pi_1$; that is, for all $\epsilon>0$,
$$\pi_1 \left[ \theta_1:\,d_{K}\{ f_0(\cdot),f_1(\cdot \mid \theta_1)\}<\epsilon \right]>0.$$
Second,
$$\sup_{\theta_2\in {\mathbb{F}_2} }\prod_{i=1}^n \frac{f_2(x_i \mid \theta_2)}{f_0(x_i)}<e^{nc}\quad\mbox{almost surely for all large   }n\mbox{     for any    }c>0.$$
Finally,
$$\inf_{\theta_2\in{\mathbb{F}_2}} d_H\{ f_2(\cdot \mid \theta_2),f_0(\cdot) \}>0.$$
Standard results  indicate that if $B_{12}\rightarrow 0$ then
$P(M_1 \mid \mbox{data})\rightarrow 0,$
whereas if $B_{12}\rightarrow\infty$ then
$P(M_1 \mid \mbox{data})\rightarrow 1.$

\subsection{Regression models: monotone vs. non-monotone}

We will now adapt this result to examine the spline-based models introduced previously. Suppose that $M_1$ is a normal model with monotone mean
function, 
$$M_1=\{N(y \mid f_1(x),\sigma_1^2),\,\pi_1(f_1,\sigma_1^2)\},$$ 
and $M_2$ is a normal model with non--monotone mean function,
$$M_2=\{N(y \mid f_2(x),\sigma_2^2),\,\pi_2(f_2,\sigma_2^2)\}.$$
Assume the $x_i$ are sampled from distribution $Q$ with support on $(0,1)$.
Let $f$ be a regression spline function, which can be monotone or non-monotone, 
and define the sieve 
$$S_n=\left\{f \in \mathcal{F}: f \,\,\mbox{has  knot points at the} \, \,  x_i \,\,\mbox{and}\,\int |f'(x)|^2\,d x <\lambda_n\right\}$$
for some $\lambda_n\uparrow+\infty$, where $\mathcal{F}$ is the space of continuous functions on $(0,1)$ with piecewise-continuous derivatives.  Thus $S_n$ includes the quadratic regression splines from Section 3.

Finally, define
$$B_{12}=\frac{\int \prod_{i=1}^n N\{y_i \mid f_1(x_i),\sigma_1^2 \} \,\pi_1(d\theta_1)  }{\int \prod_{i=1}^n N\{y_i \mid f_2(x_i),\sigma_2^2\} \,\pi_2(d\theta_2)  } \, ,$$
where $\theta_j=(f_j,\sigma_j)$.

Our main result is stated below.

\begin{theorem}
Assume that $f_0$, the true regression function, is bounded and continuous on $(0,1)$, and let $M_1$ and $M_2$ be defined as above.  Suppose that $\lambda_n=O(n^{1/4-\delta})$ for some $\delta>0$; that the $L_2$ support of both $\pi_1(f_1)$ and $\pi_2(f_2)$, which may depend on the sample size, coincide with $S_n$; and that the support of the variances coincides with $\sigma_1,\sigma_2<\sigma_+ <\infty$. Moreover, suppose that for some $\psi>0$, $\epsilon>0$, and $\rho < 1$, we have that, for all sufficiently  large $n$,
\begin{equation}
\pi^{(n)}\big\{ f:d_{\psi}(f,{\cal F}_1)<\epsilon \mid f\notin {\cal F}_1 \big\} <\rho^n \, ,\label{cond}
\end{equation}
 where
$$d_{\psi}(f,{\cal F}_1)=\inf_{f_1\in {\cal F}_1}\{d_{\psi}(f,f_1)\} \; , \quad \mbox{where} \; d_{\psi}(f,f_1)=1-n^{-1}\sum_{i=1}^n \exp[-\psi \{f(x_i)-f_1(x_i)\}^2] \, .$$
Then, if the true model lies in $M_1$,
$B_{12}$ diverges almost surely, whereas if the true model lies in  $M_2$, then $B_{12}\rightarrow 0$ almost surely.
\end{theorem}

The full proof is somewhat technical and is provided in the appendix.  The essential new ideas are conveyed by the proof of Theorem 1, while many of the technical details draw on the method of sieves studied by \citet{geman:hwang:1982}.

We now briefly describe the connection between the theorem and our spline-based approach, and the way the conditions of the theorem correspond to the three conditions of Theorem 1.  The prior for $f$, say $\pi^{(n)}$, will depend on the sampled $x_i$ and have as its actual support the functions in the sieve $S_n$. In order to ensure the true model is in the Kullback-Leibler support of the prior for all large $n$, we assume the true variance is bounded by $\sigma_+$, and that for all sufficiently large $n$, $f_0$ is in the $L_2$ support of the prior $\pi^{(n)}$.  That is, for all $\epsilon>0$, and sufficiently large $n$,
$$\pi^{(n)} \left[ f:\int \{ f_0(x)-f(x)\}^2\,Q(d x)<\epsilon \right] >0 \, .$$

In the proof of Theorem 2 we need the condition
that for some $t>0$,
$$
\int e^{t|y|}\,F_Y(y)\,dy<\infty \, ,
$$
where $F_Y$ is the true marginal distribution of the $y_i$. 
This condition on $F_Y$ and $\lambda_n$ is to ensure that the sieve maximum likelihood estimator of $f_0$ exists and converges appropriately for our purposes. Here 
$$F_Y(dy) =\int N(dy \mid f_0(x),\sigma_0^2)\,Q(d x) \, ,$$
and hence the condition on $F_Y$ will be satisfied if 
$$\int \exp\left\{f_0^2(x)\right\}\,Q(d x)<\infty \, .$$
This holds when $f_0$ is bounded, which we assume in Theorem 2.

The key to both models is whether $f\in S_n$ or not.  That is, we need 
$$\pi^{(n)}\{f\in S_n\}=1.$$
Now both models have $f$ as continuous and $f'$ as piecewise continuous, so we only need
$$\pi^{(n)}\left\{f: \int |f'|^2<\lambda_n\right\}=1.$$
This implies, for the smoothing spline model in Section 2,
$$\pi^{(n)}\left\{g:\int g^2<\lambda_n\right\}=1,$$
whereas a quadratic regression-spline model of the type described in Section 3 needs
$$\pi^{(n)}\left\{(\alpha,\beta):\int \Big| \beta_1+2\beta_2x+2\sum_{j\geq 3}\beta_j(x-\bar{x}_{j-3})_+ \Big| \mbox{d}x <\lambda_n\right\}=1.$$

Thus to ensure consistency in our proposed approaches, we need the priors for both models to be restricted so that the integrals of the squares of the gradients are bounded by an increasing function of the sample size $n$, which we denote by $\lambda_n$ and which must satisfy $\lambda_n=O(n^{1/4-\delta})$ for some $\delta>0$.  In both cases, this is straightforward to apply in practice and simply involves removing increasingly diminishing tail regions by truncating the priors.  Hence, achieving consistency does not require overly strict conditions, merely an arbitrarily small tail truncation, since for finite samples the leading constant associated with the asymptotic condition on $\lambda_n$ can be arbitrarily large.  We comment in detail on condition (\ref{cond}) in the appendix.

\section{Experiments}
\label{sec:experiments}

\subsection{Description of simulations}

This section reports the results of an extensive simulation study benchmarking the proposed methods against several alternatives that previously appeared in the literature.  There are four non-Bayesian methods in the study.  The first two are well-known in the classical literature: the U test from \citet{zheng:1996}, and the bootstrap-based test from \citet{bowman:etal:1998}.  The second two are more recent methods due to \citet{baraud:2005} and \citet{akakpo:2014}.

We also included two Bayesian methods as benchmarks. First, there is the method from \citet{salomond:2014} described earlier. Second, there is the method of Bayesian Bonferroni correction proposed by \citet{dunson:2005}, where each increment of $f$ is as
\begin{eqnarray*}
(y_i \mid f_i, \sigma^2) &\sim& \N(f_i, \sigma^2) \, , \quad f_i = f_{i-1} + \delta_i \, , \\
(\delta_i \mid w, \tau^2) &\sim& w \N^+ (\delta_i \mid 0, \tau^2)
+ (1-w) \N (\delta_i \mid 0, \tau^2) \, ,
\end{eqnarray*}
where $w$ is the mixing probability and $ \N^+$ indicates a normal distribution truncated below at zero.  By analogy with Bonferroni correction, one then chooses $w$ as a function of the sample size such that the event $\{\delta_i > 0 \ \mbox{for all} \ i\}$ has prior probability 1/2.  This model can be fit via Gibbs sampling.

Our simulations used the model $y_i = f (x_i) + \epsilon_i$, with $n = 100$ equally spaced $x$ values on $(0, 1]$, and the $\epsilon_i$ independent and identically distributed $N(0, 0.1^2)$ random variables.  We generated 100 data sets for each of the 11 test functions below:
\begin{eqnarray*}
f_1 &=& 4(x-1/2)^3 \one_{x \leq 1/2} + 0.1(x-1/2) - 0.25\exp\{ -250(x-1/4)^2\} \\
f_2 &=& -x/10 \\
f_3 &=& (-1/10)\exp\{-50(x-1/2)^2\} \\
f_4 &=& (1/10) \cos(6 \pi x) \\
f_5 &=& x/5 + f_3(x) \\
f_6 &=& x/5 + f_4(x) \\
f_7 &=& x + 1 - (1/4) \exp\{-50(x-1/2)^2\} \\
f_8 &=& x^2/2 \\
f_9 &=& 0 \\
f_{10} &=& x+1 \\
f_{11} &=& x + 1 - (9/20) \exp\{-50(x-1/2)^2\} \, .
\end{eqnarray*}
These span a variety of scenarios and include both monotone and non-monotone functions.  Functions $f_1$ to $f_7$, $f_{10}$, and $f_{11}$ were used by \citet{akakpo:2014}.  \citet{salomond:2014} used functions $f_1$ to $f_7$, and added functions $f_8$ and $f_9$.  The only difference is that we modified functions $f_1$ to $f_6$ so that they varied more slowly over the domain $[0,1]$.  This was necessary to ensure that the probability of correct classification was not essentially one.  Note that the tests in \citet{akakpo:2014} and \citet{salomond:2014} are actually for a monotone non-increasing function.  Thus for a test of a monotone non-decreasing function, we applied their tests to the negative of the simulated $y$ values from the functions above.



For the Gaussian regression spline-based test, we use the prior described in Equation (4) of the appendix.  We used $m = 33$ equally spaced knots, and set $c = 100$, $p_j = pr(\iota_j = 0) =$ 0$\cdot$8, and $q_1 = \mbox{pr}(f_m'(x) \geq 0 \mid \iota)  =$ 0$\cdot$1.  For the method-of-moments regression spline-based test, we used the prior distribution in Equation (5) of the appendix.  We set $c=10$, $p_j = pr(\iota_j = 0) =$ 0$\cdot$8, and $q_1 = \mbox{pr}(f_m'(x) \geq 0 \mid \iota)  =$ 0$\cdot$1.  Note that $c$ is a variance under the Gaussian model but a scale parameter under the method-of-moments prior, so the two values are comparable under each model.

For the constrained smoothing-spline approach, we adopt an empirical-Bayes strategy for hyperparameter specification.  We used the plug-in estimate $\sigma^2$ of the error variance derived from local linear regression under a Gaussian kernel function, with the bandwidth chosen by leave-one-out cross validation.  This is an asymptotically efficient estimator of $\sigma^2$, regardless of whether the underlying function is monotone.  We then used the mixture prior $\tau \sim (1/3) \delta_0 + (2/3) \mbox{Ga}(3,3 \hat{\tau} )$, where $\delta_0$ is a Dirac measure, and where $\hat{\tau}$ is the maximum first difference of the local-linear-regression estimate of $f(x)$.  When $\tau^2 = 0$ and $g(0) = 0$, $f(x)$ is a globally flat function.  Finally, we set $\xi \sim N(1,1)$.  Note that whenever $\xi < 0$, this implies that $f'(x)$ is an unconstrained Wiener process on $[0,1]$.  Also, observe that when $\tau^2 = 0$, our model does not depend on $\xi$.  We therefore arbitrarily set $\xi = \infty$ whenever $\tau^2 = 0$, so that $f(x)$ is classified as a monotone function.

All Markov-chain-based sampling schemes were run for a burn-in period of 20,000 iterations and a sampling period of 100,000 iterations.  The particle-filter algorithm for fitting the smoothing spline used 100,000 particles.  In each case, we calculate the Bayes factor by computing the posterior odds ratio and dividing by the original priors odds ratio.

\subsection{Results}

\begin{table}[t]
\caption{\label{tab:simresultsnew} Results of the simulation study.  Each entry is the number of times out of 100 that the function was correctly classified as either monotone or non-monotone.  Salomond: method from \citet{salomond:2014}.  Smoothing: smoothing spline test.  Gauss: regression-spline test with Gaussian priors.  MoM: regression-spline test with method-of-moments priors.  U-test: the test from \citet{zheng:1996}.  Baraud: the test from \citet{baraud:2005} with $l_n = 25$.  Akakpo: the test from \citet{akakpo:2014}.  Average: average correct classification rate across all 11 test functions.}
\vspace{0.5pc}
\centering
\begin{small}
\begin{tabular}{rrrrrrrr}
Function & Salomond & Smoothing &  Gauss &  MoM & U-test & Baraud & Akakpo \\ 
$f_1$ & 19 & 99 & 100 & 99 & 59 & 6 & 9 \\ 
$f_2$ & 83 & 72 & 74 & 63 & 59 & 64 & 33 \\ 
  $f_3$ & 51 & 34 & 35 & 49 & 59 & 53 & 43 \\ 
  $f_4$ & 73 & 80 & 91 & 98 & 0 & 92 & 92 \\ 
  $f_5$ & 56 & 95 & 85 & 90 & 99 & 24 & 25 \\ 
  $f_6$ & 86 & 96 & 99 & 100 & 34 & 77 & 74 \\ 
  $f_7$ & 13 & 92 & 91 & 47 & 16 & 1 & 4 \\ 
  $f_8$ & 98 & 80 & 93 & 93 & 41 & 100 & 100 \\ 
  $f_9$ & 96 & 98 & 95 & 95 & 99 & 97 & 94 \\ 
  $f_{10}$ & 99 & 99 & 97 & 99 & 28 & 100 & 99 \\ 
  $f_{11}$ & 100 & 100 & 99 & 99 & 100 & 71 & 82 \\
  Average & 70$\cdot$3 &  85$\cdot$9 &  87$\cdot$2 & 84$\cdot$7 & 54$\cdot$1 & 62$\cdot$3 & 59$\cdot$5
\end{tabular}
\end{small}
\end{table}

For the frequentist tests, we calculated a $p$-value under the null hypothesis of monotonicity, and rejected the null whenever $p \leq p^\star$.  For the Bayesian tests, we rejected the null hypothesis of monotonicity whenever the Bayes factor in favor of a non-monotone function exceeded a critical value $b^\star$.  To ensure a sensible comparison across these disparate methods, we calculated a separate critical value for each method as follows.  First, we simulated $1000$ test data sets of size $100$ each, for which the true regression function was $f(x) = 0$.  Thus, within each test data set, $y_i \sim \N(0, 0.1^2)$ for $i=1,\ldots, 100$.  Then, we applied each test to the 1000 simulated data sets, and calculated the threshold at which the given test would reject the null hypothesis five percent of the time.  This threshold was used as that test's critical value for all 11 test functions.  The goal of this calibration process is to ensure that all methods have an approximate size of $\alpha=$ 0$\cdot$05 under the most difficult boundary case of a globally flat function.  It also ensures that the assumed prior probability of a monotone function does not play a role in the performance of the Bayesian tests.  This calibration phase used 1000 simulations, rather than 100, because of the difficulty in estimating the tail area accurately with only 100 test statistics.

Table \ref{tab:simresultsnew} shows the number of times, out of 100 simulated data sets, that each method correctly classified functions $f_1$ to $f_{11}$ under that method's calibrated critical value.  In our simulation study, the bootstrap method of \citet{bowman:etal:1998} was uniformly beaten by the three other non-Bayesian tests.  Similarly, the Bayesian Bonferroni method was uniformly beaten by the three other Bayesian tests.   For the sake of brevity, we have not included the results for these two methods in Table \ref{tab:simresultsnew}.

As Table \ref{tab:simresultsnew} shows, although no one method is uniformly the best, the Bayesian methods exhibit the best overall frequentist properties for functions among those considered.  This is especially true of the three spline-based methods proposed here: they never perform much worse than the other methods, and in some cases perform much better than the best frequentist approach.  We conclude that the spline-based methods are a robust, powerful choice across a wide variety of underlying functions.  The method of \citet{salomond:2014} performs almost as well, and better on select cases, but is not as robust overall.  For example, neither this method nor the frequentist methods are easily able to detect the small dip in function $f_7$, which the spline-based methods can detect much more frequently.

\section{Discussion}

We conclude by highlighting one final interesting aspect of the Bayesian methods, beyond their theoretical properties and good finite-sample frequentist performance.  Unlike the classical approaches considered here, the Bayesian approaches can be used to test both global and local hypotheses about a function, without invoking any concerns about multiplicity.  We illustrate this with an example.  In Figure \ref{fig:examplef} we see two examples of data sets in our simulation study: one where $f(x) = 1 + 8.0 (x-0.75)^2$, which is non-monotone; and the other where $f(x) = 0.5x$, which is monotone.  For each of these two data sets, both the \citet{baraud:2005} and \citet{akakpo:2014} tests reject the null hypothesis of monotonicity at the $\alpha =$ 0$\cdot$05 level.  The Bayesian test using constrained smoothing splines also favors the hypothesis of non-monotonicity for each case: posterior probability 0$\cdot$98 for the function in the top pane, and 0$\cdot$74 for the one in bottom pane.

The difference between the Bayesian and frequentist methods comes when we try to identify the range of values where $f$ exhibits likely non-monotonic behavior.  Each frequentist test is an omnibus test, loosely analagous to an $F$-test in an analysis-of-variance problem.  They therefore lead to a large multiple-comparison problem if one wishes to test for specific local features of the data set that yielded a rejection from the omnibus test.  In contrast, the posterior probability associated with any such local question, such as whether the function decreases between $x=0.6$ and $x=0.8$, arises naturally from the joint posterior distribution under the Bayesian model. One may ask an unlimited number of such questions about the posterior distribution, without posing any multiplicity concerns.

\begin{figure}
\includegraphics[width=0.5\textwidth]{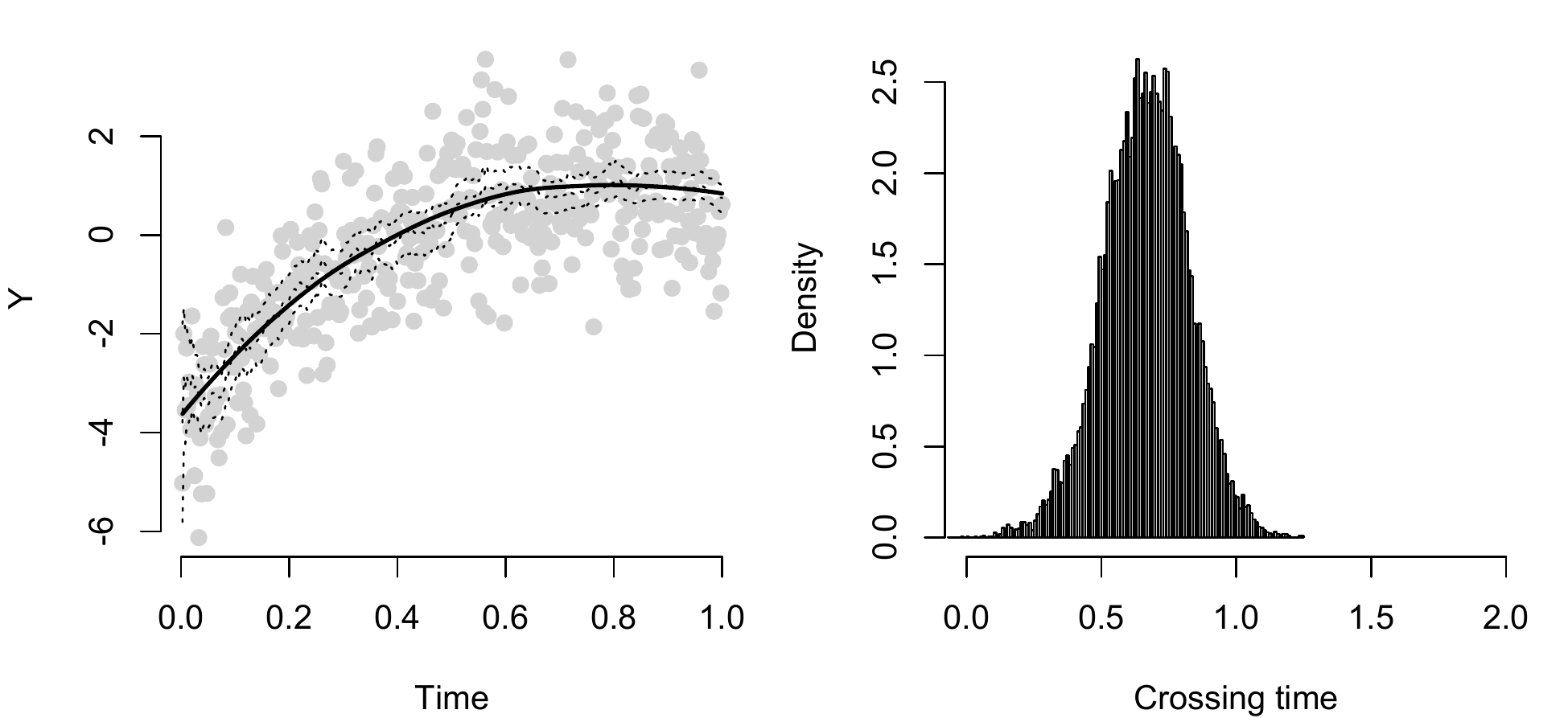} 
\includegraphics[width=0.5\textwidth]{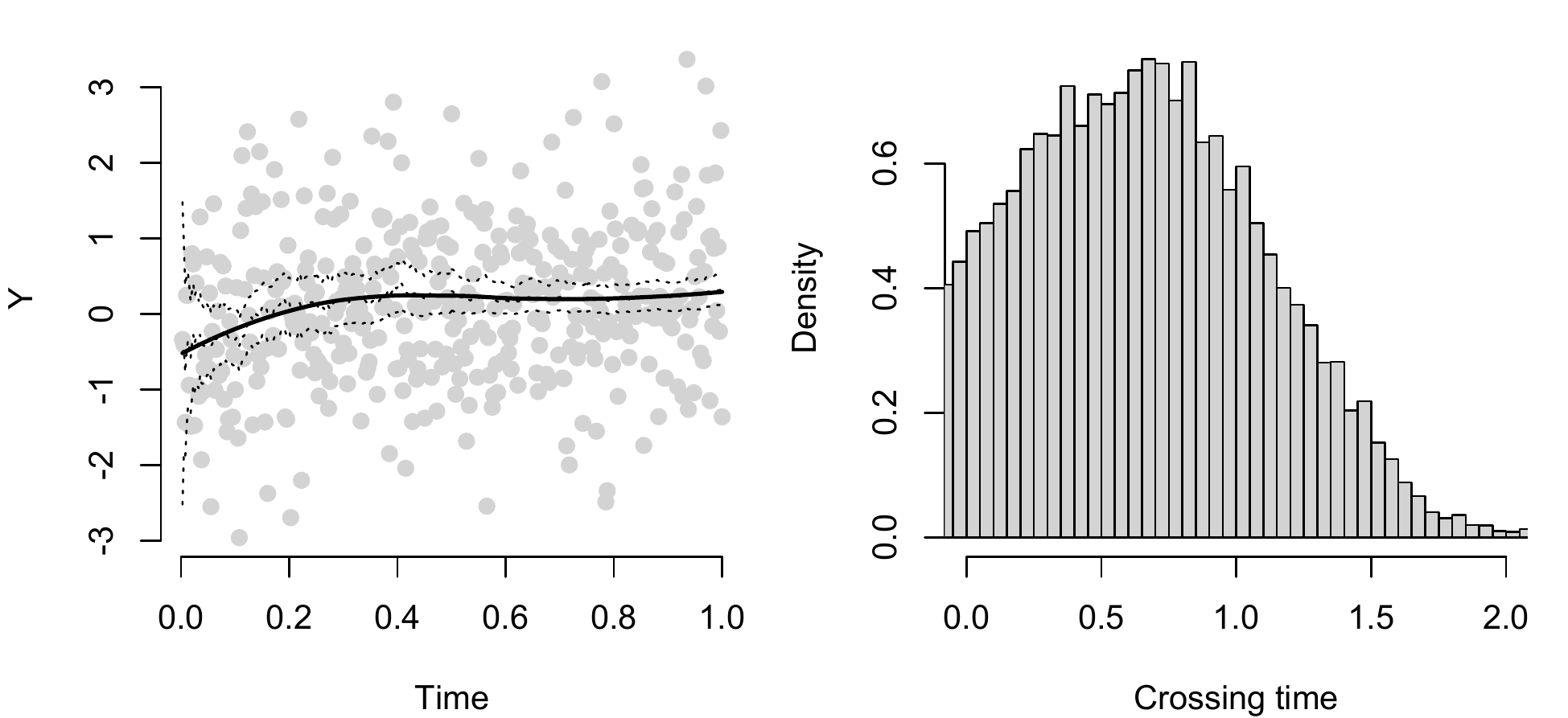}
\caption{\label{fig:examplef} The left two panels show results for a data set where $E(y) = 1 + 8.0 (x-0.75)^2$.  The first panel shows the data as grey dots, the filtered mean and $95\%$ filtered confidence interval for $f$ as dotted lines, and the smoothed mean as a solid line.  The second panel shows the histogram of draws from the posterior distribution of the first crossing time $\xi$.  The right two panels show the same two figures for a data set where $E(y) = 0.5x$.}
\end{figure}

\paragraph{Acknowledgements.}

The authors thank two referees, the associate editor, and the editor for their helpful comments in improving the paper. They also thank J.B.~Salomond for graciously sharing his computer code implementing his testing method.  The second author's research was partially supported by a CAREER grant from the U.S.~National Science Foundation (DMS-1255187).

\newpage

\appendix

\setcounter {equation} {0}

\section{Further detail on constrained smoothing splines}

Proposition 1 in the main manuscript characterizes the distribution of the increments of the stochastic process $g(s)$ for a particular value of the crossing time $\xi$.  It naturally suggests a sequential Monte Carlo algorithm for model-fitting.  Let $\Delta_j = t_{j+1} - t_j$ denote the $j$th time increment, and approximate the value of $f$ at time $t_i$ as $f(t_i) \approx \sum_{j < i} g(t_j) \Delta_j$, assuming that the time bins are sufficiently small such that $g(t_j) \Delta_j$ is a good approximation of $\int_{t_{j}}^{t_{j+1}} g(s) ds$.

Subject to this approximation,
\begin{eqnarray}
(y_i \mid f_i, \sigma^2) &\sim& \N(f_i, \sigma^2) \\
f_i &=& f_{i-1} + \Delta_{i-1} g_{i-1} \\
(g_i \mid \tau, \xi, g_{i-1}) &\sim&
\left\{
\begin{array}{l l}
\mbox{FN}(g_i \mid g_{i-1}, \xi, \tau) & \mbox{if } t_j < \xi \, ,\\
\mbox{N}(g_i \mid 0, \tau^2(t_j - \xi) ) & \mbox{if } t_{j-1} < \xi < t_j \, , \\
\mbox{N}(g_i \mid g_{i-1}, \tau^2(t_j - t_{j-1}) ) & \mbox{if } t_{j-1} > \xi  \, ,
 \end{array}
 \right. \label{eqn:propagatederiv}
\end{eqnarray}
recalling that FN denotes the fractional normal distribution from Proposition 1.  The model is characterized by the unobserved state vector $g = (g_1, \ldots, g_N)^T$, and the three unknown parameters $\xi$, $\tau$, and $\sigma^2$.  To fit it, we use the particle-filtering algorithm from \citet{liu:west:2001}, which is a modification of sequential importance resampling.  The idea is to introduce a particle approximation to the time-$t$ filtered distribution of the state $g_t$ and unknown parameters $\tau$ and $\xi$.  The problem of particle decay is handled by propagating each particle forward with a small jitter added to $\xi$ and $\tau$, drawn from Gaussian kernels centered at the current values.  This would ordinarily result in a density estimate that is over-dispersed compared to the truth.  But \citet{liu:west:2001} suggest a shrinkage correction that shifts the propagated values toward their overall sample mean, resulting in a particle approximation with the correct degree of dispersion.

With a moderate amount of data, $\sigma^2$ can be estimated quite precisely.  We use a simple plug-in estimate, derived from an unconstrained local-linear-regression estimate of $f$, with the kernel bandwidth chosen by leave-one-out cross-validation.   This produces an estimate for $\sigma^2$ that is asymptotically efficient, regardless of whether the underlying function is monotone.  We also experimented with the particle-learning approach of \citet{carvalho:etal:2010}, by tracking sufficient statistics for $\sigma^2$ as part of the state vector.  Only very small differences between these two approaches were observed, seemingly validating the simpler plug-in method.  Code implementing the method is available from the authors.

\section{Further details on constrained regression splines}

\subsection{The model}

This section develops two tests of monoticity using a regression spline model with different prior distributions for the regression coefficients.  Equation (5) from the main manuscript specifies a finitely parametrized  approximation to the function $f(x)$, and introduces our notation for the problem.  For the sake of clarity, we re-iterate this notation here.  We represent the unknown function as
$$
f_m(x) = \alpha + \beta_1 x + \beta_2 x^2 + \beta_3 (x - \xtil_1)^2_+ +  \cdots + \beta_{m+2}(x -\xtil_{m})^2_+ \, ,
$$
where the $\xtil_j$ are the ordered knot locations, and $z_+$ indicates the positive part of $z$.  Using all $m$ knots, this model may be re-written in matrix notation as $y = \alpha 1 + X \beta + \epsilon$, where $y$ is the $n\times 1$ vector of observations, $1$ is a vector of ones, $X$ is an $n \times (m+2)$ design matrix, and $\beta$ is the vector of spline coefficients.  Let $\iota$ be a vector of indicator variables whose $j$th element takes the value 0 if $\beta_j = 0$, and 1 otherwise.  Let $\beta_{\iota}$ consist of the elements of $\beta$ corresponding to those elements of $\iota$ that are equal to one, and let $p = |\iota|$ denote the number of nonzero entries in $\iota$.   \citet{shively:sager:walker:2009} derive the constraints on $\beta_{\iota}$ that ensure the monotonicity of $f_m(x)$ for a given $\iota$.  Specifically, $f_m(x)$ is monotone whenever $L_{\iota} \beta_{\iota}  \equiv \gamma_{\iota} \geq 0$, where $L_\iota$ is a known lower-triangular matrix that depends on $\iota$ and the $\xtil_j$'s.

Given $\iota$, the $\beta_\iota$ space is divided into $2^p$ disjoint regions denoted $R_\iota^{(1)}, \ldots, R_\iota^{(2^p)}$, with each region defined by a different combination of signs of the derivative $f_m'(x)$ at each of the included knots.  In general, if the sign of $f'_m(x)$ at any included knot is negative, then the function is not monotone.  Without loss of generality we may let $R^{(1)}_\iota$ denote the region where the derivative is non-negative at each of the included knots, in which case $f'_m(x) \geq 0$ for all $x \in [0,1]$.  For a specific $\iota$ and prior $p(\beta_\iota)$, one may compute the prior probability $\mbox{pr}( \beta_\iota \in R^{(1)}_\iota \mid \iota)$, which is identical to $\mbox{pr}\{ f_m'(x) \geq 0 \ \mbox{for all} \ x \in [0,1] \mid \iota \} $.

Given $\iota$ and $\sigma^2$, the prior for $\beta$ is a mixture of $2^p$ distributions, with the $d$th mixture component $\Pi_d$ constrained to have support on $R_\iota^{(d)}$.  Two specific choices for $\Pi_d$ are considered: one based on the multivariate Gaussian distribution, and the other based on the multivariate method-of-moments distribution described by \citet{johnson:rossell:2010}.  As discussed below and shown in the simulation experiment in Section 5 of the main manuscript, the two resulting tests have different small sample properties.

We first note that a multivariate Gaussian distribution for the regression coefficients is inappropriate for testing purposes, because for any given region $R_{\iota}^{(d)}$, $\mbox{pr}(\beta_{\iota} \in R_{\iota}^{(d)} \mid \iota)$ will be very small if there are a large number of knots in the model.  Because $R_{\iota}^{(1)}$ is identified with the region where $f(x)$ is monotone, this means there is a vanishingly small prior probability on the null hypothesis as the number of knots increases.

We therefore take the following approach.  For a given $\iota$ not identically 0, let $L_{\iota}$ be the lower-triangular matrix that defines the monotonicity constraint on $\beta_{\iota}$.  Let $c$ be a fixed scale factor.  The Gaussian-based prior we use is
\begin{equation}
\label{eqn:regsplinegaussprior}
(\beta_{\iota} \mid \iota, \sigma^2) \sim \sum_{d=1}^{2^p} q_d \ \mbox{TN} \left\{ \beta_{\iota} \mid 0, c \sigma^2 {(L_i L_i^T)}^{-1}, R_\iota^{(d)} \right\} \, ,
\end{equation}
where $\mbox{TN}(x \mid m, V, R)$ denotes the density function, evaluated at $x$, of the multivariate normal distribution with mean vector $m$, covariance matrix $V$, and truncation region $R$.  The integral of the density function over each of the $2^p$ mutually disjoint regions may be done analytically, due to the choice of covariance matrix.  In particular, $\mbox{pr}(\beta_{\iota} \in R_{\iota}^{(d)} \mid \iota) = q_d$, so conditional on $\iota$, $q_1$ is the prior probability of a monotone function.

As a second choice of prior, we also consider the multivariate moment-based prior,
\begin{equation}
\label{eqn:regsplineMoMprior}
(\beta_{\iota} \mid \iota, \sigma^2) \sim C_{\iota} \sum_{d=1}^{2^p} q_d \ \beta_{\iota}^T \Sigma_{\iota}^{-1} \beta_{\iota} \ \vert \Sigma_{\iota} \vert^{-1/2} \exp \left\{  - \frac{1}{2} \beta_{\iota}^T \Sigma_{\iota}^{-1} \beta_{\iota} \right\} \, ,
\end{equation}
where $C_{\iota}$ is a known normalizing constant not depending on $\beta$, and where $\Sigma_{\iota} = c \sigma^2 {(L_i L_i^T)}^{-1}$ as before.  The integral of this density function over each of the sub-regions can be done analytically to give $\mbox{pr}(\beta_{\iota} \in R_{\iota}^{(d)} \mid \iota)$.  As in the parametric linear case considered by \citet{johnson:rossell:2010}, using the moment-based prior gives a test for monotonicity with different properties than the test obtained using a Gaussian prior.  These differences are illuminated in the simulation results reported in Section 5 of the main manuscript.

To complete the model for $f_m(x)$ we must also specify priors for $\alpha$, $\sigma^2$ and $\iota$.  The intercept $\alpha$ is given a vague mean-zero normal prior with variance $10^{10}$, and the variance is given a flat prior on $[0, 10^3]$.  Using vague priors on these parameters is acceptable, as they appear in all models under consideration, and there is no indeterminacy in the Bayes factor as a result.  The prior for the knot-inclusion indicator $\iota$ was discussed in the main manuscript; each element is given a Bernoulli prior with fixed probability $p_j$.

The prior probability of monotonicity may then be computed as
$$
\mbox{pr} \{ f'_m(x) \geq 0 \; \mbox{for all} \; x \in [0,1] \} = q_1 \left[ 1 - \mbox{pr}\{\iota=(0, \ldots, 0)\}  \right] + \mbox{pr}\{\iota=(0, \ldots, 0)\} \, ,
$$
since $q_1 = \mbox{pr} \{ f'_m(x) \geq 0 \; \mbox{for all} \; x \in [0,1] \mid \iota \}$, and since a flat function with $\iota = (0, \ldots, 0)$ is also a monotone function.

\subsection{Sampling scheme for posterior inference}

To construct our Markov-chain Monte Carlo sampling scheme, we use an alternative parametrization of a regression spline model proposed in \citet{shively:sager:walker:2009}.  For a given $\iota$, define $W_{\iota} = X_{\iota} L_{\iota}^{-1}$ and $\gamma_{\iota} = L_{\iota} \beta_{\iota}$, where $L_{\iota}$ is the constraint matrix defined previously, and where $X_{\iota}$ consists of the columns of $X$ corresponding to the nonzero entries in $\iota$.  This allows us to rewrite the function as $f_m(x) = \alpha 1 + W_{\iota} \gamma_{\iota}$, and to identify the regions $R_{\iota}^{(d)}$ as orthants in the transformed $\gamma$ space.  In particular, the region of monotonocity is the first orthant, where $\gamma_{\iota} \geq 0$.  This greatly simplifies the sampling scheme.

Under this new parametrization, the Gaussian-based prior becomes
$$
(\gamma_{\iota} \mid \iota, \sigma^2) \sim \sum_{d=1}^{2^p} q_d \ \mbox{TN} \left\{ \gamma_{\iota} \mid 0, c \sigma^2 I, R_\iota^{(d)} \right\} \, ,
$$
while the moment-based prior becomes
$$
(\gamma_{\iota} \mid \iota, \sigma^2) \sim  C_{\iota} \sum_{d=1}^{2^p} q_d \ \gamma_{\iota}^T  \gamma_{\iota} \exp \left\{  - \frac{1}{2c \sigma^2} \gamma_{\iota}^T \gamma_{\iota} \right\} \, .
$$

Sampling $\alpha$ and $\sigma^2$ is straightforward, so the details of these steps are omitted.  We now discuss the details for sampling $\iota$ and $\gamma_{\iota}$ under the multivariate method-of-moments prior, the details for the Gaussian-based prior being similar. 

Let $\iota_{j}$ denote the $j$th element of $\iota$, and $\iota_{-j}$ denote the remaining $j-1$ elements.  Similarly, let $\gamma_{j}$ denote the $j$th element of $\gamma$, and $\gamma_{-j}$ the remaining elements.  We sample each $(\iota_j, \gamma_j)$ jointly, given $(\iota_{-j}, \gamma_{-j})$, the data, and all other model parameters.  To keep notation simple, we let $\Theta$ denote the complete set of other model parameters, including the entries in $\iota$ and $\gamma$ not being sampled.

We generate $(\iota_j, \gamma_j \mid y, \Theta)$ by first generating $(\iota_j \mid y, \Theta)$ marginalizing over $\gamma_j$, and then generating $(\gamma_j \mid \iota_j, y, \Theta)$.

Without loss of generality assume that we are updating $\iota_1$.  To compute $p(\iota_1 = 0 \mid y, \Theta)$, observe that if at least one element of $\iota_{-1}$ is nonzero, then
\begin{equation}
\label{eqn:piota0}
p(\iota_1 = 0 \mid y, \Theta) = C p(y \mid \iota_1 = 0, \Theta) \ p(\gamma_{-1} \mid \iota_1 = 0, \iota_{-1}) \ p(\iota_1 = 0) \, ,
\end{equation}
where $C$ is a constant.
The second term on the right-hand side is
$$
p(\gamma_{-1} \mid \iota_1 = 0, \iota_{-1}) = \tilde{q} 2^{s}  (cs)^{-1} \gamma_{-1}^T \gamma_{-1} (2 \pi c)^{-s/2} \exp \left\{  -\frac{1}{2c} \gamma_{-1}^T \gamma_{-1} \right\} \, ,
$$
where $s = | \iota _{-1} |$ is the number of nonzero elements in $\iota_{-1}$; and where $\tilde{q} = q_1$ if all elements of $\gamma_{-1}$ are positive, and $(1-q_1)/(2^s - 1)$ otherwise.  If all elements of $\iota_{-1}$ are zero, then the same representation holds with $p(\gamma_{-1} \mid \iota_1 = 0, \iota_{-1}) $ set to one.

To compute $p(\iota_1 = 1 \mid y, \Theta)$, note that
\begin{equation}
\label{eqn:piota1}
p(\iota_1 = 1 \mid y, \Theta) = C \left\{ \int_{\mathbb{R}} p(y \mid \iota_1 = 1, \gamma_1, \Theta) \ p(\gamma_1, \gamma_{-1} \mid \iota_1 = 1, \iota_{-1} ) \ d \gamma_1 \right\} p(\iota_1 = 1) \, ,
\end{equation}
where $C$ is the same constant appearing in Equation (\ref{eqn:piota0}).  Let $\delta_1 = (y - \alpha 1 - W_{-1} \gamma_{-1}) - w_1 \gamma_1$, with $w_1$ representing the first column of $W_{1, \iota_{-1}}$, and $W_{-1}$ the remaining columns.  The first term in the integrand in (\ref{eqn:piota1}) may be written as
$$
p(y \mid \iota_1 = 1, \Theta) = (2 \pi \sigma^2)^{-n/2} \exp \left\{ - \frac{1}{2 \sigma^2} \delta_1^T \delta \right\} \, .
$$
Also, $p(\gamma_1, \gamma_{-1} \mid \iota_1 = 1, \iota_{-1} ) $ is given by
\begin{equation}
\label{tomsa2}
a(\gamma_1) \ 2^r (rc)^{-1} (\gamma_1^2 + \gamma_{-1}^T \gamma_{-1}) \ (2 \pi c)^{-r/2} \ \exp \left\{ -\frac{1}{2c} (\gamma_1^2 + \gamma_{-1}^T \gamma_{-1}) \right\} \, ,
\end{equation}
where $r = s + 1$, $a(\gamma_1) = \tilde{q}$ if $\gamma_1 > 0$, and $(1-q_1)/(2^r -1)$ if  $\gamma_1 < 0$.  Recall that  $\tilde{q} = q_1$ if all elements of $\gamma_{-1}$ are positive, and $(1-q_1)/(2^r - 1)$ otherwise.

Let $\tilde{y} = y - \alpha 1 - W_{-1} \gamma_{-1}$, and let
$$
d = 2^r (cr)^{-1} c^{-1/2} (2 \pi c)^{-s/2}  (2 \pi \sigma^2)^{-n/2} \exp \left( - \frac{1}{2c} \gamma_{-1}^T \gamma_{-1} \right) \left( \frac{w_1^T w_1}{\sigma^2} + c^{-1} \right)^{-1/2} p(\iota_1 = 1) \, . 
$$
Then (\ref{eqn:piota1}) can be written as
\begin{eqnarray}
p(\iota_1 = 1 \mid y, \Theta) &=& C d \exp \left[ - \frac{1}{2\sigma^2} \left\{ \frac{ \tilde{y}^T \tilde{y} - (w_t^T \tilde{y})^2} { w_1^T w_1 + \sigma^2/c} \right\} \right] \nonumber \\
&\times& \left\{ \frac{1-q_1}{2^r-1}
\int_{-\infty}^0 h(\gamma_1)  \ d \gamma_1
+ \tilde{q} \int_0^{\infty} h(\gamma_1) \ d \gamma_1 \right \} \, ,
\end{eqnarray}
where
\begin{eqnarray*}
h(\gamma_1) &=& (\gamma_1^2 + \gamma_{-1}^T \gamma_{-1})(2 \pi \tau^2_{\gamma})^{-1/2} \exp \left\{ - \frac{1}{2 \tau^2_{\gamma}}(\gamma_1 - \mu_{\gamma})^2  \right\} \\
\hat{\gamma}_1 &=& (w_1 ^T w_1)^{-1} w_1^T \tilde{y} \\
\mu_{\gamma} &=& (w_1^T w_1 + \sigma^2/c)^{-1} w_t^T \tilde{y} \\
\tau^2_{\gamma} &=& \sigma^2 (w_1^T w_1 + \sigma^2/c)^{-1} \, .
\end{eqnarray*}
The two integrals with respect to $\gamma_1$ may be done analytically.

If $\iota_1 = 0$, then $\gamma_1$ need not be generated.  If $\iota_1 = 1$, then
$$
p(\gamma_1 \mid y, \iota_1 = 1, \Theta) \propto p(y \mid \iota_1 = 1, \gamma_1, \Theta) \ p(\gamma_1, \gamma_{-1} \mid \iota) \, ,
$$
with $p(\gamma_1, \gamma_{-1} \mid \iota)$ already given in (\ref{tomsa2}).  This is a mixture distribution that can be sampled by an efficient accept/reject algorithm.

The details for the Gaussian prior for very similar, with the modification that the term $\gamma_1^2 + \gamma_{-1}^T \gamma_{-1}$ does not appear in front of the normal kernel, and that drawing from the mixture distribution above may be done using constrained Gaussian distributions, without appealing to an accept/reject algorithm.

\section{Proof of Theorem 2}

\subsection{Notation and preliminary lemma}

Recall that the two models under comparison are $M_1$, a normal model with monotone mean
function, 
$$M_1=\{N(y \mid f_1(x),\sigma_1^2),\,\pi_1(f_1,\sigma_1^2)\} \, ,$$ 
and $M_2$, a normal model with non--monotone mean function,
$$M_2=\{N(y \mid f_2(x),\sigma_2^2),\,\pi_2(f_2,\sigma_2^2)\} \, ,$$
where $\pi_1$ and $\pi_2$ are the priors over the monotone and non-monotone functions, respectively.  We assume the $x_i$ are sampled from distribution $Q$ with support on $(0,1)$.

Let $f$ be a regression spline function, which can be monotone or non-monotone.  In the main paper, we defined the sieve $S_n$ as
$$S_n=\left\{f \in \mathcal{F}: f \,\,\mbox{has  knot points at the} \, \,  x_i \,\,\mbox{and}\,\int |f'(x)|^2\,d x <\lambda_n\right\}$$
for some $\lambda_n\uparrow+\infty$, where $\mathcal{F}$ is the space of continuous functions on $(0,1)$ with piecewise-continuous derivatives.  In addition, we must also define the sieve 
$$S_n'=\left\{f \in \mathcal{F}': f \,\,\mbox{has  knot points at the}\,(x_i)\,\,\mbox{and}\,\int |f'(x)|^2\,d x <\lambda_n\right\}$$
for some $\lambda_n\uparrow+\infty$, where $\mathcal{F}'$ is the space of continuous, piecewise linear functions on $(0,1)$.  Clearly $S_n' \subset S_n$, and we assume that all functions are defined on $S_n$.  But on pages 10 and 11 and in Theorem 2 of a Brown University technical report in 1981 by S.~Geman on sieves for nonparametric estimation of densities and regressions, it is shown that $\widehat{f}(x)$, the minimizer of
$$\sum_{i=1}^n \{y_i-f(x_i)\}^2,$$
subject to $f\in S_n$, is piecewise linear, with knots at the $(x_i)$.  Therefore, in what follows, all maximum-likelihood estimates are assumed without loss of generality to be restricted to $S_n'$, even though the priors have support over the larger sieve.

Let $f_0$ be the true unknown function.   Let $\widehat{f}_1$ be the maximum likelihood estimate of $f_0$ over $\mathcal{F}_1$, and $\widehat{f}_2$ the maximum-likelihood estimate of $f_0$ over $\mathcal{F}_2$, both restricted to the sieve $S_n'$.  We write $\widehat{f}$ to denote the unrestricted maximum-likelihood estimate over $S_n'$.  We start with a Lemma which provides a useful result for both $f_0\in {\cal F}_1$ and $f_0\in {\cal F}_2$.

\vspace{0.2in}
\noindent
{\sc Lemma 1.} Let $\widehat{f}(x)$ be the maximum likelihood estimator of the regression spline, including both monotone and non-monotone parts of $S_n'$.  If
\begin{equation}
n^{-1}\sum_{i=1}^n \{\widehat{f}(x_i)-f_0(x_i)\}^2\rightarrow 0\quad\mbox{almost surely}, \label{eq3}
\end{equation}
then 
$$\lim\inf_n\,\inf_{f_1 \in \mathcal{F}_1} n^{-1}\sum_{i=1}^n (y_i-f_1(x_i))^2\geq \sigma_0^2\,\,\mbox{almost surely} \, ,$$
and
$$\lim\inf_n\,\inf_{f_2 \in \mathcal{F}_2} n^{-1}\sum_{i=1}^n (y_i-f_2(x_i))^2\geq \sigma_0^2\,\,\mbox{almost surely} \, .$$

\begin{proof}
If (\ref{eq3}) holds, then from the triangular inequality, each of the following holds almost surely under $f_0$.  First,
$$n^{-1}\sum_{i=1}^n (y_i-f_0(x_i))^2\leq n^{-1}\sum_{i=1}^n (y_i-\widehat{f}(x_i))^2+n^{-1}\sum_{i=1}^n (\widehat{f}(x_i)-f_0(x_i))^2 \quad \mbox{a.s.}$$
The last term goes to 0, and so in the limit,
$$n^{-1}\sum_{i=1}^n (y_i-f_0(x_i))^2\leq n^{-1}\sum_{i=1}^n (y_i-\widehat{f}(x_i))^2 \, . $$
But from the definition of the maximum likelihood estimator,
$$n^{-1}\sum_{i=1}^n (y_i-\widehat{f}(x_i))^2\leq n^{-1}\sum_{i=1}^n (y_i-f_0(x_i))^2 \, ,$$
and hence, as
$$n^{-1}\sum_{i=1}^n (y_i-f_0(x_i))^2\rightarrow \sigma_0^2 \, ,$$
it follows that
$$n^{-1}\sum_{i=1}^n (y_i-\widehat{f}(x_i))^2\rightarrow \sigma_0^2 \, .$$
To elaborate on this, let 
$$a_n=n^{-1}\sum_{i=1}^n (y_i-f_0(x_i))^2,\quad b_n=n^{-1}\sum_{i=1}^n (y_i-\widehat{f}(x_i))^2\quad\mbox{and}\quad c_n=n^{-1}\sum_{i=1}^n (\widehat{f}(x_i)-f_0(x_i))^2.$$
So
$$a_n-c_n\leq b_n\leq a_n$$
and $c_n\rightarrow 0$ and $a_n\rightarrow\sigma_0^2$; so also $b_n\rightarrow\sigma^2_0$.

Then, if $f_0$ is monotone, we have by definition
$$n^{-1}\sum_{i=1}^n (y_i-\widehat{f}(x_i))^2\leq n^{-1}\sum_{i=1}^n (y_i-\widehat{f}_2(x_i))^2,$$
and, if $f_0$ is non-montone, we have also by definition
$$n^{-1}\sum_{i=1}^n (y_i-\widehat{f}(x_i))^2\leq n^{-1}\sum_{i=1}^n (y_i-\widehat{f}_1(x_i))^2,$$
completing the proof.
\end{proof}

\subsection{Proof of theorem}
Conditions under which (\ref{eq3}) holds are to be found, for example, in \citet{geman:hwang:1982}, using sieves.  These are easily verified under the assumptions of our theorem.  To see this, note that if 
$$\int \exp(t|y|)\,F_Y(dy)<\infty$$
for some $t>0$, which we assume to be the case, then the condition
$\lambda_n=O(n^{1/4-\delta})$ for some $\delta>0$ is sufficient to ensure that
$$\int_0^1 \left\{ \widehat{f}(x)-f_0(x)\right\}^2\,Q(dx)\rightarrow 0\quad\mbox{almost surely.} $$
It is then easy to show, assuming $f_0$ is bounded,  that
$$\int_0^1 \left\{ \widehat{f}(x)-f_0(x)\right\}^2\,Q_n(dx)\rightarrow 0\quad\mbox{almost surely,}$$
where $Q_n$ is the empirical distribution of the $(x_i)$. Hence, (\ref{eq3}) follows.

Having established that (\ref{eq3}) holds, we now take the two cases in turn: first where $f_0$ is non-monotone, and then where $f_0$ is monotone.
\paragraph{$f_0$ is non-monotone.}
Let us first assume that $f_0(x)$, the true mean function on $\mathbb{X}$, is non--monotone. That is, the true density is $g_0(y \mid x,\theta_0)=N(y \mid f_0(x),\sigma_0^2)$.
So consider
$$B_{12}=\frac{\int \prod_{i=1}^n N(y_i \mid f_1(x_i),\sigma_1^2)\,\pi_1(d\theta_1)  }{\int \prod_{i=1}^n N(y_i \mid f_2(x_i),\sigma_2^2)\,\pi_2(d\theta_2)  }.$$
The denominator, with the additional factor involving the true model, is given by
$$I_{n2}=\int \prod_{i=1}^n \frac{N(y_i \mid f_2(x_i),\sigma_2^2)}{N(y_i \mid f_0(x_i),\sigma_0^2)}\,\pi_2(d\theta_2) . $$
The theorem assumes that the prior for $(f_2,\sigma_2^2)$ has $f_0$ in its Kullback--Leibler support:
$$\pi_2 \left[ (f_2,\sigma_2): \int d_K\{ N(\cdot|f_2(x),\sigma_2^2),N(\cdot|f_0(x),\sigma_0^2) \} \,Q(d x)<\epsilon \right]>0$$
for all $\epsilon>0$.  Under this condition, $I_{n2}>e^{-n\tau}$ almost surely for all large $n$, for any $\tau>0$.

The numerator can be written as
$$I_{n1}=\int \prod_{i=1}^n \frac{N(y_i \mid f_1(x_i),\sigma_1^2)}{N(y_i \mid f_0(x_i),\sigma_0^2)}\,\pi_1(d\theta_1) $$
and $f_1$ belongs to the set of monotone functions.
We need to show that $I_{n1}<e^{-n\delta}$ almost surely for all large $n$ for some $\delta>0$.  This can be shown with the following two conditions.

\begin{enumerate}
\item The follow bound holds almost surely under $f_0$:
$$\lim\inf_n\,\inf_{f_1 \in \mathcal{F}_1} n^{-1}\sum_{i=1}^n \{y_i-f_1(x_i)\}^2\geq \sigma_0^2 \, .$$
\item If the $x$ are sampled from $Q$ then, for some constant $\psi>0$,
$$\sup_{f_1} \int \exp\left\{-\psi(f_0(x)-f_1(x))^2\right\}\,Q(d x)<1.$$
\end{enumerate}
Under the assumptions of the theorem, the first condition is ensured by Lemma 1, and the second condition holds because $f_0(x)$ is a fixed non--monotone function, and it is not possible for a monotone function to get arbitrarily close to it.

We use these conditions to establish the result for $f_0$ non-monotone as follows.  First,
$$I_{n1}\leq \left\{\sup_{f_1,\sigma_1^2}\prod_{i=1}^n \frac{N(y_i|f_1(x_i),\sigma_1^2)}{N(y_i|f_0(x_i),\sigma_0^2)}\right\}^{1/2}\,
\int \left\{\prod_{i=1}^n\frac{N(y_i|f_1(x_i),\sigma_1^2)}{N(y_i|f_0(x_i),\sigma_0^2)}\right\}^{1/2}\,\pi_1(d\theta_2).$$
Write these two terms as $K_{n1}$ and $J_{n1}$, respectively.

If we let $\widehat{f}_1(x_i)$ be the minimizer of
$$\sum_{i=1}^n \{y_i-f_1(x_i)\}^2,$$
then the appropriate $\widehat{\sigma_1^2}$ is given by
$$n^{-1}\sum_{i=1}^n\{y_i-\widehat{f}_1(x_i)\}^2.$$
Hence,
$$K_{n1}=\left( \frac{\sigma_0^2}{n^{-1}\sum_{i=1}^n(y_i-\widehat{f}_1(x_i))^2 } \right)^{n/2}\,\exp\left\{-\half n+\half n\,n^{-1}\sum_{i=1}^n(y_i-f_0(x_i))^2/\sigma_0^2\right\}.$$
Since
$$n^{-1}\sum_{i=1}^n (y_i-f_0(x_i))^2/\sigma_0^2\rightarrow 1\quad\mbox{almost surely},$$
and using condition 1, it follows that $K_{n1}<e^{n\eta}$ almost surely for all large $n$, for any $\eta>0$.

On the other hand, the expectation of $J_{n1}$ is given by
$$E (J_{n1})=\int \left[ 1-\half \int d_H^2\{N(\cdot \mid f_1(x),\sigma_1^2),N(\cdot \mid f_0(x),\sigma_0^2)\}\,Q(d x) \right] ^n\,\pi_1(d\theta_1).$$
Now
$$\begin{array}{ll}
\half d_H^2\{ N(\cdot \mid f_1(x), \sigma_1^2), N(\cdot \mid f_0(x),\sigma_0^2)\}  & =1-\sqrt{\frac{2\sigma_0\sigma_1}{\sigma_1^2+\sigma_0^2}}
\exp\left[-\frac{\{f_0(x)-f_1(x)\}^2}{4(\sigma_0^2+\sigma_1^2)} \right] \\ \\
& \geq 1-\exp\left[-\psi\,\{f_0(x)-f_1(x)\}^2 \right]\end{array}$$
for some constant $\psi>0$, if we impose an upper bound on the prior for $\sigma_1^2$.

Therefore, with condition 2, $E(J_{n1})<e^{-n\kappa}$, for some $\kappa>0$, and hence $J_{n1}<e^{-n\phi}$ almost surely for all large $n$, for some $\phi>0$. This result follows from an application of the 
Markov inequality and the Borel-Cantelli Lemma.  We have
$$P(J_{n1}>e^{-n\phi})<e^{n\phi}e^{-n\kappa}$$
and so for $\kappa>\phi$,
$$\sum_n P(J_{n1}>e^{-n\phi})<\infty.$$
From Borel-Cantelli, this implies that
$$J_{n1}<e^{-n\phi}$$
almost surely for all sufficiently large $n$.
Hence, $I_{n1}<e^{n\eta}e^{-n\phi}$ almost surely for all large $n$ for any $\eta>0$ and for some $\phi>0$, yielding $I_{n1}<e^{-n\delta}$ almost surely for all large $n$ for some $\delta>0$. Putting all this together, we have
$$B_{12}\leq \exp\{n(-\delta+\tau)\}$$
almost surely for all large $n$, for any $\tau>0$, and for some $\delta>0$. So choose $\tau<\delta$ to get the required result
that $B_{12}\rightarrow 0$ almost surely as $n\rightarrow \infty$.

\paragraph{$f_0$ monotone.}
Now let us consider the reverse case when $f_0$ is a fixed monotone function.  As before, we reason from two conditions, which are the mirror of conditions 1 and 2 above.
\begin{description}

\item 3. The following bound holds almost surely under the true model:
$$\lim\inf_n\,\inf_{f_2 \in \mathcal{F}_2} n^{-1}\sum_{i=1}^n \{y_i-f_2(x_i)\}^2\geq \sigma_0^2 \, .$$

\item 4. For some constant $\psi>0$ and $\rho_0<1$,
\begin{equation}
\int \left[\int \exp\{-\psi(f(x)-f_0(x))^2\}\,Q(d x)\right]^n\,\pi_2(d f)<\rho_0^n. \label{eq1}
\end{equation}
for all sufficiently large $n$.
\end{description}

Condition 3 follows from Lemma 1.  To understand condition 4, let $f_2$ be any element of $\mathcal{F}_2$.  It cannot be the case that at all points $f_2$ is arbitrarily close to $f_0$, as $f_2$ is not monotone.  Yet $f_2$ can be arbitrarily close to $f_0$ for most points and differ only in one interval of the sieve, which is the most difficult case to detect.  Condition 4 is necessary to guarantee sufficient average separation of $f_2$ and $f_0$, under the prior $\pi_2(f_2)$.

We therefore now rely on establishing Condition 4 by showing that a suitably constructed prior will not put sufficiently large mass on $f_2$ being arbitrarily close to any fixed function, such as $f_0$. We will then go on to use Condition 4 to prove the result.  To this end, we  can split the outer integral into two parts: one with
$$\int \exp\{-\psi(f(x)-f_0(x))^2\}\,Q(d x)<1-\epsilon \, ,$$
and the other with
$$\int \exp\{-\psi(f(x)-f_0(x))^2\}\,Q(d x)>1-\epsilon$$
for a sufficiently small $\epsilon$.
The case for the inner integral being bounded above by $1-\epsilon$ is clear, so we need to  show that
\begin{equation}
\pi_2\left\{f_2: \int \exp\left\{-\psi(f_0(x)-f_2(x))^2\right\}\,Q(d x)>1-\epsilon\right\}<\rho^n, \label{eq2}
\end{equation}
for some $\rho<1$.
Now (\ref{eq2}) is of the form
\begin{equation}
\pi^{(n)} \left\{ d(f,{\cal F}_0)<\epsilon \mid f\notin {\cal F}_0 \right\} <\rho^n \, ,\label{cond}
\end{equation}
where
$$d(f,{\cal F}_0)=\inf_{f_0\in {\cal F}_0}\{d(f,f_0)\}$$
and
$$d(f,f_0)=1- \int\exp\left\{-\psi(f_0(x)-f_2(x))^2\right\}\,Q(d x).$$
This border region between ${\cal F}_1$ and ${\cal F}_2$ is known to be problematic and putting exponentially small mass there is one solution. For example, Salomond solves this problem by modifying the test to
$$H_0: \tilde{d}(f,{\cal F}_0)<\epsilon\quad\mbox{vs}\quad H_1:\tilde{d}(f,{\cal F}_0)>\epsilon$$
for some alternative distance $\tilde{d}$.

In a simple example, suppose we have $f(x)=\beta x$; then we can ensure (\ref{eq2}) by taking
$$\pi^{(n)}\big(\beta\in(-\delta,0)\big)<\rho^n$$
for some $\delta>0$.

To investigate (\ref{eq2}) a little further, 
suppose we have $Q(d x)$ which takes $n$ equi--spaced samples $(x_1,\ldots,x_n)$ and, for each $f_2$, define
$k$ as the number of points for which $|f_2(x_i)-f_0(x_i)|<\delta$ for some arbitrarily small $\delta>0$. The condition now translates to
$$\pi_2(k/n>1-\delta^*)<\rho^n$$
for some $\delta^*>0$. To see this, and letting $z_i=\psi (f_0(x_i)-f_2(x_i))^2$, we need to consider
$$n^{-1}\sum_{i=1}^n e^{-z_i}>1-\epsilon.$$
Putting $\phi_i=1-e^{-z_i}$, we need to consider
$$n^{-1}\sum_{i=1}^n \phi_i<\epsilon \, ,$$
and $z_i<\delta$ implies $\phi_i<\delta$. Now putting $\delta=M\epsilon$, for some $M>1$, then we need at least $k$ of the $\phi_i$ to be less than $\delta$ where
$$n^{-1}(n-k)M\epsilon<\epsilon.$$
Thus, $k>n(1-1/M)$, i.e. $1/M=\delta^*$.

Hence, the prior $\pi_2$ can not put a function $f_2$  too close to any monotone function, including $f_0$, for a proportion $n(1-1/M)$ of the $n$ points, with mass larger than $\rho^n$ for some $\rho<1$.
For each $i$, consider the distribution of $f_2(x_i)$ given $\{f_2(x_{i-1}),\ldots,f_2(x_{1})\}$. From the prior construction we can ensure there is a maximum value for 
$$\pi_2 \left\{  (|f_2(x_i)-f_0(x_i)|<\delta  \mid f_2(x_{i-1}),\ldots,f_2(x_1) \right\} ,$$
i.e. this probablity has an upper bound, say $\rho_1<1$. This means that with the same positive, but arbitrarily small probability, $|f_2(x_i)-f_2(x_{i-1})|$ is larger than $2\delta$ for all $i$. 

Hence,
$$P \left\{ k>n(1-\delta^*) \right\} \leq P \left\{ k_u>n(1-\delta^*) \right\}$$
where $k_u$ is binomial with parameters $\rho_1$ and $n$.  This is because there are $n+1$ intervals, and the function only needs to be decreasing on one of them to be a non-monotone function. Here, $P\left\{ k_u>n(1-\delta^*) \right\}$ is simply introduced as an upper bound.

Therefore,  using a normal approximation to the binomial,
$$k_u/n\approx N\{ \rho_1,\rho_1(1-\rho_1)/n \} \, ,$$
and hence
$$P\{ k>n(1-\delta^*) \} <P\{ k_u>n(1-\delta^*)\} \approx P\left(z>\xi\sqrt{n}\right)$$
for some $\xi>0$, and where $z$ is a standard normal random variable.  Using the asymptotic expression for the survival function of the normal, we have
$$P\left\{ k_u>n(1-\delta^*) \right\} \approx \frac{c_1}{\sqrt{n}}\exp(-c_2 n)$$
for positive constants $c_1$ and $c_2$. Hence, (\ref{eq2}) holds true, and so (\ref{eq1}) holds true.

We could modify the priors described in the paper to adhere to (\ref{cond}).  However, in practice it would make no discernible difference, since $\rho, \psi$ and $\delta$ are all arbitrary constants and can be taken arbitrarily small or large as required.

\end{spacing}

\singlespace
\begin{footnotesize}
\bibliographystyle{abbrvnat}
\bibliography{masterbib,business_articles}
\end{footnotesize}

\end{document}